\documentclass[5p, times, twocolumn,preprint, authoryear]{elsarticle}

\usepackage{hyperref}
\usepackage[hyperref]{xcolor}
\hypersetup{
  colorlinks   = true, 
  urlcolor     = red, 
  linkcolor    = red, 
  citecolor   = red 
}
\usepackage{amsthm}
\theoremstyle{plain}
\newtheorem{thm}{Theorem}
\theoremstyle{definition}
\newtheorem{defn}{Definition}
\theoremstyle{remark}

\theoremstyle{definition}

\theoremstyle{plain}
\newtheorem{Lemma}{Lemma}
\theoremstyle{plain}
\newtheorem{Claim}{Claim}
\theoremstyle{definition}
\newtheorem{Sketch of proof}{Sketch of proof}
 \usepackage{color}
\usepackage{braket}
\usepackage[ampersand]{easylist}
\usepackage{dsfont}
\usepackage{mathrsfs}
\usepackage{multirow}
\usepackage{braket}
\usepackage{dsfont}
\usepackage{mathrsfs}
\usepackage{natbib} 
\biboptions{longnamesfirst, semicolon}
\usepackage{rotating}
\usepackage{pifont}
\usepackage{amsmath}
\usepackage{multicol}
\usepackage{amssymb}

\begin{document}

\begin{frontmatter}

\title{Power Spectrum Identification for Quantum Linear Systems}

\author[*]{Matthew Levitt}
\ead{pmxml2@nottingham.ac.uk}
\author[*]{M\u{a}d\u{a}lin Gu\c{t}\u{a}}
\author[**]{Hendra I. Nurdin}
\address[*]{School of Mathematical Sciences, University of Nottingham, University Park, NG7 2RD Nottingham, UK}
\address[**]{School of Electrical Engineering and Telecommunications, UNSW Australia, Sydney NSW 2052, Australia}


\begin{abstract}
In this paper we investigate system identification for general quantum linear systems. We consider the situation where the input field  is prepared as stationary (squeezed) quantum noise. In this regime the output field is characterised by the power spectrum, which encodes covariance of the output state. We address the following two questions: (1) Which parameters can be identified from the power spectrum?  (2) How to construct a system realisation from the power spectrum? The power spectrum depends on the system parameters via the transfer function. We show that the transfer function can be uniquely recovered from the power spectrum, so that equivalent systems are related by a symplectic transformation. 
\end{abstract}


\end{frontmatter}

\section{Introduction}

 \begin{figure}[h]
\centering
\includegraphics[scale=0.30]{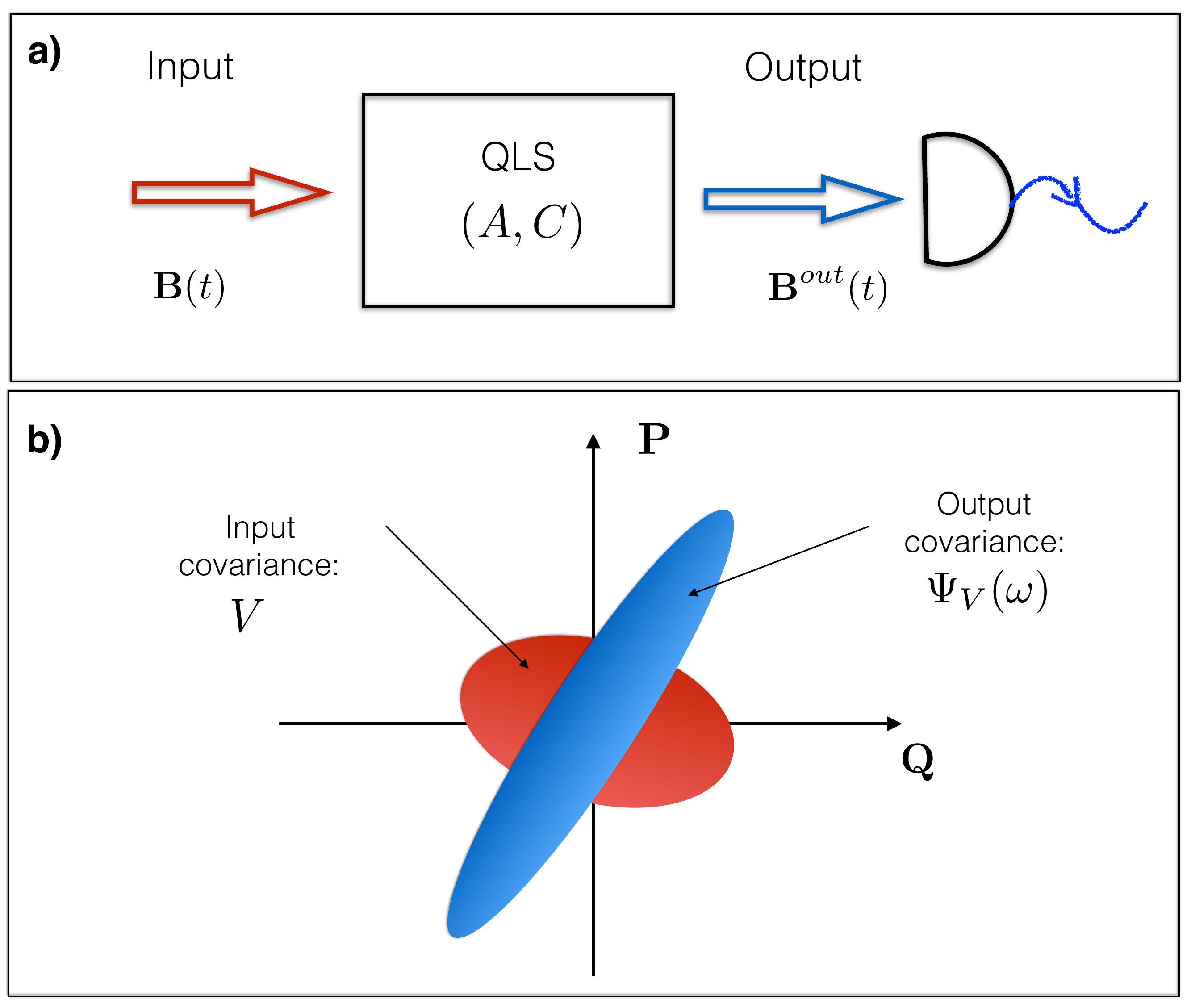}
\caption{a) System identification problem: find parameters $(A,C)$ of a linear input-output system by measuring output. b) Stationary scenario: \textit{power spectrum} describes output covariance in frequency domain.\label{twoapproach}}
\end{figure}

System identification theory \cite{Ljung,Pintelon&Schoukens, GUTA3, GUTA4} lies at the interface between control theory and statistical inference, and deals with the estimation of unknown parameters of dynamical systems and processes from input-output data. 
The integration of control and identification techniques plays an important role e.g. in adaptive control \cite{Astrom}.

In this paper we consider system identification for \textit{quantum  linear systems} (QLSs).  QLSs are a class of models used in quantum optics, opto-mechanical systems, electro-dynamical  systems,  cavity QED systems  and elsewhere  \cite{Naoki, wall, Tian, GZ, STOCK, DOHERTY}. They have many applications, such as quantum memories, entanglement generation, quantum information processing and quantum control \cite{memory, memory2, JNP, nurdin2009, MILB, LB2, Dong}. The  framework required to describe these is the celebrated \textit{quantum stochastic calculus} \cite{path2}. 

Quantum linear systems are examples of input-output models (see Figure \ref{twoapproach}). Typically, one has access to the field and is able to prepare a time-dependent input. After the coupling, the parameters of the system (Ôblack-boxÕ) are imprinted on the output.
In a nutshell, the system identification problem is to estimate dynamical parameters from the output  data, obtained by performing measurements on the output fields. The identification of linear systems is by now a well developed subject in `classical' systems theory \cite{glover, kalman, Ljung, Ho, anders, Youla, zhou, Pintelon&Schoukens, Davies}, but has not been fully explored in the quantum domain \cite{Guta, Guta2}.  We distinguish two contrasting approaches to the identification of QLSs 

In the first approach, one probes the system with a known \emph{time-dependent} input signal (e.g. coherent state), then uses the output measurement data to compute an estimator of the dynamical parameter(s). In this setting  the \textit{transfer function} entirely encapsulates  the systems input-output behaviour. Therefore, the basic identifiability problem is to find the class of systems with the same transfer function. This problem has been addressed, firstly for the special class of \textit{passive} QLSs in \cite{Guta} and then for general QLSs in \cite{levitt}. In particular, it was seen that  \textit{minimal} systems with the same transfer function are related by symplectic transformations on the space of system modes. 

The second approach and the one we consider here is to probe the systems with \textit{time-stationary} pure gaussian states with independent increments (see Figure  \ref{twoapproach}), i.e., squeezed vacuum noise. If the system is minimal and Hurwitz stable, the dynamics exhibits an initial transience period after which it reaches stationarity and the output is in a stationary Gaussian state, whose covariance in the frequency domain is  given by the \textit{power spectrum}. 
The power spectrum depends quadratically on the transfer function, so the parameters which are identifiable in the  stationary scenario will also be identifiable in the time-dependent one. Our goal is to understand to what extent the converse is also true. 
This problem is of the type: `for a square rational matrix $V(s)$, where $s\in\mathbb{C}$ find rational matrix $W(s)$ such that 
$$V(s)=W(s)W(-\overline{s})^{\dag}$$ for all $s\in\mathbb{C}$, which in the classical literature is called the \textit{spectral factorisation problem} \cite{anders}.
Note that our previous work \cite{levitt} looked at this problem for a generic class of single input single output (SISO) QLSs. Now, for a given minimal system there may exist lower dimensional systems with the same power spectrum. To understand this, consider the system's stationary state and note 
that it can be uniquely written as a tensor product between a pure and a mixed Gaussian state 
(cf. the symplectic decomposition \cite{WOLF}. It is known \cite{levitt}   that by restricting the system to the mixed component leaves the power spectrum unchanged. Conversely, if the stationary state is fully mixed, there exists no smaller dimensional system with the same power spectrum. Such systems will be called \emph{globally minimal}, and can be seen as the analogue of minimal systems for the stationary setting.

The main result here is to show that under global minimality the power spectrum determines the transfer function, and therefore the equivalence classes are the same as those in the transfer function. It is interesting to note that this equivalence is a consequence of the unitarity and purity of the input state, and does not hold for a generic classical linear system \cite{anders, glover}. 
The key to our proof is in reducing the power spectrum identifiability problem to an equivalent  transfer function identifiability problem. 

This paper is organised as follows: In Section \ref{red1} we  review the setup of input-output QLS, and their associated transfer function. In Section \ref{section3} we outline the power spectrum identifiability problem. We introduce the notion of global minimality for systems with minimal dimension for a given power spectrum and review recent important results. 
Our main identifiability result is presented in  Section \ref{section4}, cf., Theorem \ref{main}. Finally, we outline a method to construct a globally minimal system realisation  from the power spectrum.

\subsection{Preliminaries and notation}

We use the following notations: 
For a matrix $X=(X_{ij})$ the following symbols:
$\overline{X}=(X_{ij}^*)$,  $X^{T}=(X_{ji})$, $X^{\dag}=(X_{ji}^*)$
represent the complex conjugation, transpose and adjoint matrix respectively, where `*' indicates complex conjugation. We also use the `doubled-up notation' $\breve{X}:=\left[X^T, \overline{X}^T\right]^T$ and $\Delta(A, B):=\left[A, B; \overline{B}, \overline{A}\right]$. 
For a matrix $Z\in\mathbb{R}^{2n\times2m}$ define $Z^\flat=J_mZ^{\dag}J_n$, where $J_n=\left[\mathds{1}_n, 0; 0, -\mathds{1}\right]$.  A similar notation is used for matrices of operators. 
We use `$\mathds{1}$' to represent the identity matrix or operator. $\delta_{jk}$ is Kronecker delta and $\delta(t)$ is Dirac delta.
The commutator is denoted by $[\cdot, \cdot]$.

\begin{defn}\label{def.symplectic}
A matrix $S\in\mathbb{C}^{2m\times 2m}$ is said to be \textit{$\flat$-unitary} if it is invertible and satisfies
\[S^\flat S=SS^\flat =\mathds{1}_{2m}.\]

If additionally, $S$ is of the form $S=\Delta(S_-, S_+)$ for some $S_-, S_+\in\mathbb{C}^{m\times m}$ then we say that it is \textit{symplectic}. \end{defn}

\section{Quantum Linear Systems}\label{red1}

In this section we briefly review the QLS theory, highlighting along 
the way results that will be relevant for this paper. 
We refer to \cite{GZ} for a more detailed discussion on the input-output 
formalism, and  to the review papers \cite{peter, path, path2, nurdin} 
for the theory of linear systems.

\subsection{Time-domain representation}\label{timed}

A linear input-output quantum system is defined as a continuous 
variables (cv) system coupled to a Bosonic environment, such that their joint 
evolution is linear in all canonical variables. 
The system is described by the column vector of annihilation operators, 
$\mathbf{a}:=[\mathbf{a}_1,\mathbf{a}_2, \dots, \mathbf{a}_n]^T$, 
representing the $n$ cv modes. 
Together with their respective creation operators 
$\mathbf{a}^{*}:=[\mathbf{a}^{*}_1,\mathbf{a}^{*}_2, \dots, \mathbf{a}^{*}_n]^T$ 
they satisfy the canonical commutation relations (CCR)
$
\left[\mathbf{a}_i, \mathbf{a}^{*}_j\right]=\delta_{i j}\mathds{1}.
$
We denote by $\mathcal{H}:= L^2(\mathbb{R}^n)$ the Hilbert space of the system carrying the standard representation of the 
$n$ modes. The environment is modelled by $m$ bosonic fields, called \textit{input channels}, 
whose fundamental variables are the fields 
$\mathbf{B}(t):=\left[\mathbf{B}_{1}(t), \mathbf{B}_{2}(t), \ldots, 
\mathbf{B}_{m}(t)\right]^T$, where $t\in \mathbb{R}$ represents time. 
The fields satisfy the CCR 
\begin{eqnarray}
\left[\mathbf{B}_{i}(t), \mathbf{B}^*_{j}(s)\right]=\mathrm{min}\{t,s\}\delta_{ij}\mathds{1}.
\end{eqnarray}
Equivalently, this can be written as 
$\left[\mathbf{b}_{i}(t),\mathbf{b}_{j}^*(s)\right]
=\delta(t-s)\delta_{ij}\mathds{1}$, where 
$\mathbf{b}_{i}(t)$ are the infinitesimal (white noise) annihilation operators formally defined as 
$\mathbf{b}_{i}(t):=\mathrm{d}\mathbf{B}_{i}(t)/\mathrm{d}t$ \cite{peter}. 
The operators can be defined in a standard fashion on the Fock space 
$\mathcal{F}= \mathcal{F}(L^2(\mathbb{R})\otimes \mathbb{C}^m)$ \cite{LB}. We consider the scenario where the input is prepared in a \emph{pure, stationary in time, mean-zero, Gaussian state} with independent increments characterised by the covariance matrix 
\begin{align}\label{Ito}
\left<\begin{smallmatrix}\mathrm{d}\mathbf{B}(t)\mathrm{d}\mathbf{B}(t)^{\dag}&\mathrm{d}\mathbf{B}(t)\mathrm{d}\mathbf{B}(t)^T\\\mathrm{d}\mathbf{B}^*(t)\mathrm{d}\mathbf{B}(t)^{\dag}&\mathrm{d}\mathbf{B}^*(t)\mathrm{d}\mathbf{B}(t)^T\end{smallmatrix}\right>=\left(\begin{smallmatrix} N^T+\mathds{1}&M\\M^{\dag}&N\end{smallmatrix}\right)\mathrm{d}t:={V}(N,M)\mathrm{d}t,
\end{align}
where the brackets denote a quantum expectation. 
Note that $N=N^{\dag}$, $M=M^T$ and $V\geq0$,   which ensures that the state does not violate the uncertainty principle. 
In particular, $N=M=0$ corresponds to the  vacuum state, while pure squeezed states  satisfy $\overline{M}(N+1)^{-1}M=N$ \cite{squeezing}. 


The dynamics of a general input-output system is determined by the system's Hamiltonian and its coupling to the environment. In the Markov approximation, the joint unitary evolution of system and environment is described by the (interaction picture) unitary ${\bf U}(t)$ on the joint space 
$\mathcal{H}\otimes \mathcal{F}$, which is the solution of the quantum stochastic differential equation  \cite{LB,Dong,GZ,path,path2}
\begin{eqnarray}
\label{eq.QSDE} 
&&\mathrm{d}\mathbf{U}(t) 
         :=\mathbf{U}(t+\mathrm{d}t)-\mathbf{U}(t) \\
          &&=\left(-\left(\mathbf{K}+i \mathbf{H}\right)\mathrm{d}t             + \mathbf{L}\mathrm{d}\mathbf{B}^{\dag}
               - \mathbf{L}^{\dag}\mathrm{d}\mathbf{B}
               -\frac{1}{2}\mathbf{L}^{\dag}\mathbf{L}\mathrm{d}t\right)\mathbf{U}(t),\nonumber
\end{eqnarray}               
%
where  $\mathbf{K}=\frac{1}{2}\left(\mathbf{L}^{\dag}(1+N^T)\mathbf{L}+\mathbf{L}^TN\overline{\mathbf{L}}- \mathbf{L}^{\dag}M\overline{\mathbf{L}}-\mathbf{L}^T\overline{M}\mathbf{L}\right)$ \cite{squin} and initial condition ${\bf U}(0)= \mathbf{I}$.
Here, ${\bf H}$ and ${\bf L}$ are system operators describing the 
system's Hamiltonian and the coupling to the fields; 
$\mathrm{d}{\bf B}_i(t), \mathrm{d}{\bf B}_i^*(t)$, 
are increments of fundamental quantum stochastic processes describing the 
creation and annihilation operators in the  input channels.  

For the special case of \emph{linear} systems, 
the coupling and Hamiltonian operators are of the form  
\begin{eqnarray*}
\mathbf{L} &=& C_{-}\mathbf{a}+C_{+}\mathbf{a}^*,\\
\mathbf{H}&=&  \mathbf{a}^{\dag}\Omega_{-}\mathbf{a}+\frac{1}{2}\mathbf{a}^T\Omega_{+}^{\dag}\mathbf{a}+\frac{1}{2}\mathbf{a}^{\dag}\Omega_{+}\mathbf{a}^*,
\end{eqnarray*}
for $m\times n$ matrices $C_{-}, C_{+}$ and $n\times n$ matrices $\Omega_{-}, \Omega_{+}$ satisfying $\Omega_{-}=\Omega_{-}^{\dag}$ and $\Omega_{+}=\Omega_{+}^T$. 
As shown below, this insures that all canonical variables evolve linearly in time. Indeed, let $\mathbf{a}(t)$ and 
$\mathbf{B}^{out}(t)$ be the Heisenberg evolved system and output variables 
\begin{eqnarray}
       \mathbf{a}(t):=\mathbf{U}(t)^{\dag}\mathbf{a}\mathbf{U}(t),~~~ 
       \mathbf{B}^{out}(t):=\mathbf{U}(t)^{\dag}\mathbf{B}(t)\mathbf{U}(t).
\end{eqnarray}
By using the QSDE \eqref{eq.QSDE} and the Ito rules (\ref{Ito}) one can obtain the 
following Ito-form quantum stochastic differential equation of the QLS in the doubled-up notation 
\begin{eqnarray}\label{langevin}
       \mathrm{d}\breve{\bf a}(t) &=& 
            A \breve{\bf a}(t)\mathrm{d}t-C^{\flat}\mathrm{d} \breve{\bf B}(t),\\
       \mathrm{d} \breve{\bf B}^{out}(t) &=& 
            C \breve{\bf a}(t)\mathrm{d}t+\mathrm{d} \breve{\bf B}(t),
\end{eqnarray}
where $ \breve{\bf a} := ({\bf a}^T, {\bf a^*}^T)^T$,  $C:=\Delta\left(C_{-}, C_{+}\right)$ and $A:=\Delta\left(A_{-}, A_{+}\right)=-\frac{1}{2}C^\flat C-iJ_n\Omega$ with $\Omega= \Delta\left(\Omega_{-}, \Omega_{+}\right)$ and
\[
A_{\mp}:=-\frac{1}{2}\left(C_{-}^{\dag}C_{\mp}-C_{+}^{T}\overline{C}_{\pm}\right)-i\Omega_{\mp}.
\] 
To be explicit, the behaviour of the linear system is completely characterised by the dynamical parameters $(C, A)$ (or equivalently $(C, \Omega)$).
Note that not all choices of  $A$ and $C$ may be physically realisable as open quantum systems \cite{JNP}.

A special case of linear systems is that  of \textit{passive} quantum linear systems (PQLSs) \cite{Guta} for which $C_{+}=0$ and 
$\Omega_{+}=0$.

\subsection{Controllability and observability}
By taking the expectation with respect to the initial joint system state of Equations (\ref{langevin}) we obtain the following classical linear system
\begin{eqnarray}\label{classicallangevin}
       \mathrm{d}\left<\breve{\bf a}(t) \right>= 
            A\left< \breve{\bf a}(t) \right>\mathrm{d}t-C^{\flat}\mathrm{d}\left<\breve{\bf B}(t)\right>,\\
       \mathrm{d}\left<\breve{\bf B}^{out}(t)\right> = 
            C\left<\breve{\bf a}(t) \right>\mathrm{d}t+\mathrm{d}\left<\breve{\bf B}(t)\right>.
\end{eqnarray}

\begin{defn}
The quantum linear system (\ref{langevin}) is said to be Hurwitz stable (respectively controllable, observable) if the corresponding classical system (\ref{classicallangevin}) is Hurwitz stable (respectively controllable, observable).
\end{defn}
In general, for a quantum linear system observability and controllability are equivalent \cite{Indep}. A system possessing one (and hence both) of these properties is called \textit{minimal}.  However, although the statement   [Hurwitz $\implies$ minimal]  is true \cite{Naoki}, the converse statement ([minimal $\implies$ Hurwitz]) is not necessarily so \cite{levitt}. We therefore  assume that all systems considered here are Hurwitz (hence minimal). 

\subsection{Frequency-domain representation}\label{freqrep}

For linear systems it is often useful to switch from the time domain dynamics 
described above, to the frequency domain picture. 
Recall that the Laplace transform of a generic process ${\bf x}(t)$ is defined by
\begin{equation}\label{eq.laplace.classic}
 \mathbf{x}(s) := \mathcal{L}[\mathbf{x}](s)
      =\int_{-\infty}^\infty e^{-st}{\bf x}(t)dt, 
\end{equation}
where $s\in\mathbb{C}$. 
In the Laplace domain the input and output fields are related as follows 
\cite{Yanagisawa}:
\begin{equation}
\breve{\bf b}^{out}(s) = 
\Xi(s)
 \breve{\bf b}(s),
\label{iol}
\end{equation}
where $ \Xi(s)$ is \emph{transfer function matrix} of the system%
\begin{equation}
\label{eq.transfer.function.general}
       \Xi(s)=\Big\{\mathds{1}_m-C(s\mathds{1}_n-A)^{-1}C^{\flat}\Big\}=
       \left(\begin{smallmatrix}
       \Xi_{-}(s)&\Xi_{+}(s)
       \\
       \overline{\Xi_{+}\left(\overline{s}\right)}& \overline{\Xi_{-}\left(\overline{s}\right)}
       \end{smallmatrix}\right).
\end{equation}

In particular, the frequency domain input-output 
relation is
$\breve{\bf b}^{out}(-i\omega) = 
\Xi(-i\omega)
 \breve{\bf b}(-i\omega).$
The corresponding commutation relations are
$\left[\mathbf{b}(-i\omega),\mathbf{b}(-i\omega')^*\right]
=i\delta(\omega-\omega')\mathds{1}$, and similarly for the output modes\footnote{Note that the position of the conjugation sign is important here because in general  $\mathbf{b}(-i\omega')^*$ and $\mathbf{b}^*(-i\omega')$ are not the same, cf. equation (\ref{eq.laplace.classic}).}. 
As a consequence, the transfer matrix $\Xi(-i\omega)$ is symplectic for all frequencies $\omega$. 

We do not consider static squeezing or scattering processes on the field in this paper (see e.g. \cite{squeezing}).



\subsection{Transfer function identifiability}\label{identifiability}
The input-output relation \eqref{iol} shows that the experimenter can at most  
identify the transfer function $\Xi(s)$ of the system with any measurement of the field.
The following result from \cite{levitt} tells us which dynamical parameters of a QLS can be identified by observing the output fields for appropriately chosen input states. 
\begin{thm}\label{symplecticequivalence}
Let $\left(A, C \right)$ and $\left(A', C'\right)$ be two minimal, and stable QLSs. 
Then they have the same transfer function if and only if there exists a symplectic matrix $T$ such that 
\begin{equation}\label{eq.equivalene.classes}
A'=TAT^{\flat}, \,\,\, C'=CT^{\flat}.
\end{equation}
\end{thm}
Therefore, without any additional information, we can at most identify the symplectic equivalence class of systems here. 





\section{Power spectrum identification; problem formulation}\label{section3}


We consider a setting where the input fields are stationary 
 `squeezed quantum noise', i.e. a zero-mean, pure Gaussian state with time-independent increments, which is completely characterised by its covariance matrix $V$, cf. equation \eqref{Ito}. In the frequency domain the state can be seen as a continuous tensor product over frequency modes of squeezed states  with covariance ${ V}$. Since we deal with a linear system, the input-output map consists of applying a (frequency dependent) unitary Bogolubov transformation whose linear symplectic action on the frequency modes is given by the transfer function 
$$
\breve{\bf b}^{out} (-i\omega)  =  
\Xi (-i\omega) \breve{\bf b} (-i\omega).
$$
Consequently, the output state is a Gaussian state consisting of independent frequency modes with covariance matrix
\begin{eqnarray*} 
\left< \breve{\bf b}^{out} (-i\omega) \breve{\bf b}^{out} (-i\omega^\prime)^\dagger \right>=   
 \Psi_{ V}(-i\omega) \delta(\omega-\omega^\prime)  
 \end{eqnarray*}
 where $ \Psi_{V}(-i\omega)$ is the restriction to the imaginary axis of the \emph{power spectral density} (or power spectrum) defined in the Laplace domain by
\begin{equation}\label{powers}
\Psi_{V}(s)= \Xi(s){ V}\Xi(-\overline{s})^{\dag}.
\end{equation}

Our goal is to find which system parameters are identifiable from the field,   where the quantum input has a given covariance matrix ${ V}$. Since in this case the output is uniquely defined by its power spectrum 
$\Psi_{ V}(s)$ this reduces to identifying the equivalence class of systems with a given power spectrum. 
Moreover, since the latter depends on the system parameters via the transfer function, it is clear that one can identify `at most as much as' the transfer function discussed in Section \ref{identifiability}. In other words  
the corresponding equivalence classes are at least as large as those described by symplectic transformations \eqref{eq.equivalene.classes}. 
%

In the analogous classical problem, the power spectrum can also be computed from the output correlations. 
The spectral factorisation problem \cite{Youla} is tasked with finding a transfer function  from the power spectrum. There are known algorithms \cite{Youla, Davies} to do this. One then finds a system realisation (i.e. matrices governing the system dynamics) for the given transfer function \cite{Ljung}. The problem is that the map from power spectrum to transfer functions is  non-unique, and each factorisation could lead to system realisations of differing dimension. For this reason, the concept of \textit{global minimality} was introduced in \cite{kalman} to select the transfer function with smallest system dimension. This raises the following question: is global minimality sufficient to uniquely identify the transfer function from the power spectrum? 
The answer is in general negative \footnote{However, under the assumption of \textit{outer} transfer functions this identification  is unique (see \cite{nerve}).} (see \cite{anders, glover, nerve}). Our aim is to address these questions in the quantum case. Note that these questions  have been answered for a generic class of SISO systems in \cite{levitt} by using a brute force argument to identify the poles and zeros of the transfer function from those of the power spectrum. 

We conclude this section by formally introducing  global minimality and describing two results  that will be useful later.

\begin{defn}
A system $(A,C)$ is said to be globally minimal for (pure) input covariance $V$ if there exists no lower dimensional system with the same power spectrum, $\Psi_{V}$.
\end{defn}

For example, if the input is the vacuum and the system is passive, then the  power spectrum will be vacuum, which is  the same as that of a zero-dimensional system. 

Observe that  as the input is pure, we may write it as 
$V=SV_{\mathrm{vac}}S^{\dag}$ for some symplectic matrix $S$, where $V_{\mathrm{vac}}=\left(\begin{smallmatrix}1&0\\0&0\end{smallmatrix}\right)$.
Specifically, 
\begin{equation}\label{vtrick}
S=\Delta\left((N^T+1)^{1/2},M\left(N^{\dag}+1\right)^{-1/2}\right)
. 
\end{equation}
Now,  since input is  known (i.e the choice of the experimenter) we  instead consider the modified system with coupling and hamiltonian operators $\tilde{C}:=CS^{\flat}$ and $\tilde{\Omega}=\Omega$, which has  power spectrum $\tilde{\Psi}(s)=S^{\flat}\Psi(s)\left(S^{\dag}\right)^{\flat}$. In this basis the field is in vacuum.
In light of this we  will  assume  that the input is vacuum. 




The following  theorem from 
 \cite{levitt} links global minimality with the purity of the stationary state. 
 \begin{thm}\label{equivalence} 
 Let $\mathcal{G}:= \left(S, C,\Omega\right)$ be a QLS with input 
 $V_{\mathrm{vac}}$.

1. The system  is globally minimal  if and only if the (Gaussian) stationary state of the system with covariance $P$ satisfying the Lyapunov equation 

\begin{equation}\label{leap}
AP+PA^{\dag}+C^{\flat}{ V_{\mathrm{vac}}}\left(C^{\flat}\right)^{\dag}=0
\end{equation}
is fully mixed.

2. A non-globally minimal system is the series product of its restriction to the pure component and the mixed component. 

3. The reduction to the mixed component is globally minimal and has the same power spectrum as the original system.   

\end{thm}


\begin{Lemma}\label{LEM1}
Suppose that we have a QLS $\left(C, \Omega\right)$ with input $V_{\mathrm{vac}}$, then the following are equivalent:
\begin{enumerate}
\item The system is  globally minimal
 \item $\left(A, C^{\flat}{ V}_{\mathrm{vac}}\right)$ is controllable.
 \item  $\left({ V}_{\mathrm{vac}}C, A^{\flat}\right)$ is observable.
 \end{enumerate}
 \end{Lemma}
 
\begin{proof}
For the equivalence between (1) and (2): Using Theorem \ref{equivalence}, global minimality is equivalent to a fully mixed stationary state, which is in turn equivalent to  $P>0$ in \eqref{leap}.  Furthermore, by Theorem 3.1 in \cite{zhou}  $P>0$ in \eqref{leap} is equivalent to  $\left(A, C^{\flat}V_{\mathrm{vac}}\right)$ being controllable. 

It remains to show equivalence between (2) and (3). Firstly, by the duality condition  in \cite[Theorem 3.3]{zhou}  $\left(A, C^{\flat}{ V}_{\mathrm{vac}}\right)$  controllable is equivalent to   $\left({V}_{\mathrm{vac}}\left(C^{\flat}\right)^{\dag}, A^{\dag}\right)$ observable. It therefore remains to show equivalence between the observability of $\left({V}_{\mathrm{vac}}\left(C^{\flat}\right)^{\dag}, A^{\dag}\right)$ and $\left({V}_{\mathrm{vac}}C, A^{\flat}\right)$. 

Suppose that $\left({ V}_{\mathrm{vac}}\left(C^{\flat}\right)^{\dag}, A^{\dag}\right)$ is observable. To show observability of $\left({V}_{\mathrm{vac}}C, A^{\flat}\right)$ we need to show that for all eigenvectors and eigenvalues of $A^{\flat}$, i.e. $A^{\flat}y=\lambda y$, then ${ V}_{\mathrm{vac}}Cy\neq0$ \cite{zhou}. 
To this end suppose that $A^{\flat}y=\lambda y$, then $A^{\dag}\left(Jy\right)=\lambda\left(Jy\right)$, which by the observability of $\left({ V}_{\mathrm{vac}}\left(C^{\flat}\right)^{\dag}, A^{\dag}\right)$ implies that ${ V}_{\mathrm{vac}}\left(C^{\flat}\right)^{\dag}\left(Jy\right)\neq0$. Therefore, ${ V}_{\mathrm{vac}}Cy\neq0$ and we are done.
The reverse implication follows similarly.
\end{proof}




\section{Power spectrum identifiability}\label{section4}

 In this section we show that  two globally minimal systems have the same power spectrum iff they have the same transfer function. We show this by treating the power spectrum of the quantum system as a transfer function of a cascade of two classical systems (with the combined system having twice as many modes). 
 We then solve the equivalent minimal transfer function problem, which  is much simpler than the original problem.

\subsection{Description of power spectrum as a cascade of systems}\label{god}

Using \eqref{powers},  write the power spectrum as a transfer function of the following two cascaded \cite{zhou} systems:
\begin{itemize}
\item The first system  is $\left(-A^{\flat}, -C^{\flat}, -V_{\mathrm{vac}}C, V_{\mathrm{vac}}\right)$
\item The second system is  $\left(A, -C^{\flat}V_{\mathrm{vac}}, C, V_{\mathrm{vac}}\right)$. 
\end{itemize} 
It should be understood that the first system is fed into the second (see Fig \ref{casc}). Note that the first system is  unstable, whereas the second is stable. A representation for the resultant cascaded system with transfer function $\Psi(s)J$ is \cite{zhou}
\begin{equation}\label{cask}
\left(\tilde{A}, \tilde{B}, \tilde{C}, \tilde{D}\right):=
\left(\left(\begin{smallmatrix} -A^{\flat}&0\\C^{\flat}V_{\mathrm{vac}}C&A\end{smallmatrix}\right), \left(\begin{smallmatrix} -C^{\flat}\\-C^{\flat} V_{\mathrm{vac}}\end{smallmatrix}\right), \left(\begin{smallmatrix} -V_{\mathrm{vac}}C& C\end{smallmatrix}\right), V_{\mathrm{vac}}       \right).
\end{equation}

Now, in in the form \eqref{cask} notice that $\tilde{A}$ has $4n$ eigenvalues. It is also 
 lower block triangular (LBT) with the following properties:
\begin{itemize}
\item[1)] \label{pil1} It has $2n$ right-(generalised\footnote{A matrix is diagonalisable iff it has a full basis of eigenvectors. Generalised eigenvectors are a next best thing to eigenvectors enabling one to `almost diagonalise' a matrix. More specifically, a vector  $x$ is a generalised eigenvector of rank $m$ with corresponding eigenvalue $\lambda$ if 
$$\left(A-\lambda \mathds{1}\right)^mx=0$$
(but
$\left(A-\lambda \mathds{1}\right)^{m-1}x\neq0$).
For every matrix $A$ there exists an invertible matrix $M$, whose columns consist of the generalised eigenvectors, such that $J=M^{-1}AM$ where $J$ is a matrix called the \textit{Jordan normal matrix} and is given by 
$$J=\mathrm{diag}(J_1, J_2,..., J_r) \quad \mathrm{where} \quad
J_i=\left(\begin{smallmatrix} \lambda_i &1&&\\
&\lambda_i &1&\\
&&\ddots&1\\
&&&\lambda_i\end{smallmatrix}\right).$$
 })-eigenvectors of the form $\left(\begin{smallmatrix}0\\y_2^{(i)}\end{smallmatrix}\right)$ with (possibly non-distinct) eigenvalues $\lambda^{(i)}$, which satisfy  $\mathrm{Re}(\lambda^{(i)})<0$. Note that $y_2^{(i)} $ and $\lambda^{(i)}$ are right-(generalised) eigenvectors and eigenvalues of $A$.
\item[2)] \label{pil2} 
It has $2n$ left-(generalised) eigenvectors of the form  $\left(\begin{smallmatrix}x_1^{(i)},&0\end{smallmatrix}\right)$ with  (possibly non-distinct) eigenvalues $\mu^{(i)}$, which satisfy $\mathrm{Re}(\mu^{(i)})>0$. Note that $x_1^{(i)}$ and $\mu^{(i)}$ are left-eigenvectors and eigenvalues of $-A^{\flat}$.
\end{itemize}

\begin{defn}
A matrix $A$ is called \textit{proper ordered lower block triangular (proper LBT)} if it is it LBT and satisfies 1) and 2).
\end{defn}

 \begin{figure}[h]
\centering
\includegraphics[scale=0.30]{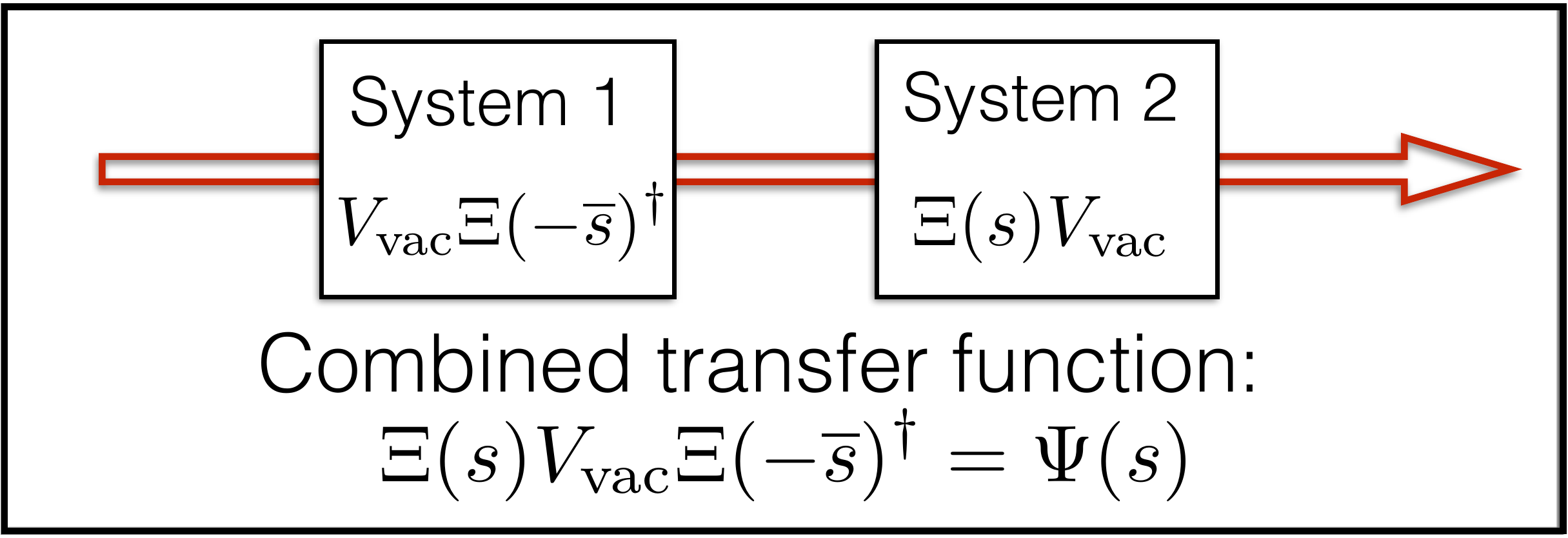}
\caption{The setup in Section \ref{god} where the power spectrum is treated as two systems connected in series. \label{casc}}
\end{figure}



\begin{Lemma}\label{hud}
If two proper LBT matrices, $\tilde{A}$ and $\tilde{A}'$,  are related via $\tilde{A}'=T\tilde{A}T^{-1}$, where $T$ is invertible, then $T$ is LBT. 
\end{Lemma}

The proof is in \ref{goh}.
The final result of this subsection
 will be key to our identifiability result later.

\begin{thm}\label{games}
The quantum system $(C, \Omega)$ is globally minimal if and only if the system \eqref{cask} is minimal. 
\end{thm}

\begin{proof}
The reverse implication here is trivial. For the `if' statement we need to prove controllability and observability. 

Firstly, the observability of $\left(\tilde{C}, \tilde{A}\right)$. Suppose that  
\begin{equation}\label{contro} 
\left(\begin{smallmatrix} -A^{\flat}&0\\C^{\flat}V_{\mathrm{vac}}C&A\end{smallmatrix}\right)\left(\begin{smallmatrix}y_1\\y_2\end{smallmatrix}\right)=\left(\begin{smallmatrix}\lambda y_1\\\lambda y_2\end{smallmatrix}\right),
\end{equation}
then in order to show observability we require that $\left(\begin{smallmatrix} -V_{\mathrm{vac}}C& C\end{smallmatrix}\right)\left(\begin{smallmatrix}y_1\\y_2\end{smallmatrix}\right)\neq0$. There are two cases; either $y_1=0$ or $y_1\neq0$. 
\begin{itemize}
\item If $y_1=0$ then \eqref{contro} reduces to   $Ay_2=\lambda y_2$ and so the observability of $A$  tells us that $Cy_2\neq0$. Hence 
$\left(\begin{smallmatrix} -V_{\mathrm{vac}}C& C\end{smallmatrix}\right)\left(\begin{smallmatrix}0\\y_2\end{smallmatrix}\right)\neq0$. 
\item For $y_1\neq0$, the proof is a little trickier.  Suppose to the contrary that the system is not observable. That is, there exists a vector $\left(\begin{smallmatrix}y_1\\y_2\end{smallmatrix}\right)$ satisfying \eqref{contro} such that 
\begin{equation}\label{ps}
V_{\mathrm{vac}}Cy_1=Cy_2
\end{equation} 
Firstly, from 
 \eqref{contro} it is clear that  $-A^{\flat}y_1=\lambda y_1$, hence  $V_{\mathrm{vac}}Cy_1\neq0$ by global minimality (Lemma \ref{LEM1}). 
 We also have $C^{\flat}V_{\mathrm{vac}}Cy_1+Ay_2=\lambda y_2$  from \eqref{contro}, hence $-A^{\flat}y_2=\lambda y_2$ using \eqref{ps}. 
 On the other hand, letting  $y_2=\left(\begin{smallmatrix}u_1\\u_2\end{smallmatrix}\right)$, where $u_1, u_2$ are $n$ dimensional complex vectors, then by the doubled-up properties of $A^{\flat}$ it follows that $\left(\begin{smallmatrix}\overline{u}_2\\\overline{u}_1\end{smallmatrix}\right)$ is  also an eigenvector of $-A^{\flat}$ (with eigenvalue $\overline{\lambda}$).
 Therefore, $V_{\mathrm{vac}}C\left(\begin{smallmatrix}\overline{u}_2\\\overline{u}_1\end{smallmatrix}\right)\neq0$ by global minimality (Lemma \ref{LEM1}). Finally, this condition implies that $\overline{C_-}u_2+\overline{C_+}u_1\neq0$, which is a contradiction to \eqref{ps}.
  Hence the system is observable.




\end{itemize}

Showing controllability of $\left(\tilde{A}, \tilde{B}\right)$ can be achieved by similar means. Alternatively, we can use the dual properties of observability and controllability to show this. To this end, in order to show that $\left(\tilde{A}, \tilde{B}\right)$ is controllable it is enough to show that $\left(\tilde{B}^{\dag}, \tilde{A}^{\dag}\right)$ is observable \cite{zhou}.
In light of this, suppose that $\tilde{A}^{\dag}\left(\begin{smallmatrix}z_1\\z_2\end{smallmatrix}\right)=\lambda\left(\begin{smallmatrix}z_1\\z_2\end{smallmatrix}\right)$, which, by using the definition of $\tilde{A}$, is equivalent to 
\begin{align*}
-JAJz_1+C^{\dag}V_{\mathrm{vac}}CJz_2&=\lambda z_1 \quad\mathrm{and}\quad
A^{\dag}z_2=\lambda z_2.
\end{align*}
These equations can be written in matrix form as
%
$$
\tilde{A}\left(\begin{smallmatrix}Jz_2\\-Jz_1\end{smallmatrix}\right)=-\lambda\left(\begin{smallmatrix}Jz_2\\-Jz_1\end{smallmatrix}\right).
$$
Now, because $\left(\tilde{C}, \tilde{A}\right)$ is observable, it follows that
$$-C\left(Jz_1\right)-V_{\mathrm{vac}}C\left(Jz_2\right)\neq0.$$
This condition is equivalent to $\tilde{B}^{\dag}\left(\begin{smallmatrix}z_1\\z_2\end{smallmatrix}\right)\neq0$. 
\end{proof}


\subsection{Main result}

\begin{thm}\label{main}
Let $\left(C_1, \Omega_1\right)$ and $\left(C_2, \Omega_2\right)$ be two globally minimal  and stable QLSs for input $V_{\mathrm{vac}}$,
%
then
$$
\Psi_{1}(s)=\Psi_{2}(s) \,\, for~all~s \quad \Leftrightarrow \quad \Xi_1(s)=\Xi_2(s)\,\, for~all~s 
$$
\end{thm}

\begin{proof}\label{sofa}
Firstly, by Theorem \ref{games} the system \eqref{cask} is minimal. Therefore, from the classical literature  transfer function equivalent systems are related via 
\begin{align}
\label{c1}&\tilde{A}'=T\tilde{A}T^{-1}, \quad\tilde{B}'=T\tilde{B}, \quad\tilde{C}'=\tilde{C}T^{-1}, \quad \tilde{D}'=\tilde{D}.
\end{align}
Moreover, its observability and controllability matrices, $\mathcal{O}$ and $\mathcal{C}$, will have full rank. 
Additionally, by Lemma \ref{hud} such a similarity transformation must be lower block triangular. 

Now writing $T$ as 
$$\left(\begin{smallmatrix} T_1 & 0\\ T_3 & T_4\end{smallmatrix}\right),$$ to complete the proof it remains to show that (a) $T_3=0$, (b) $T_1=T_4$, (b) $T_1^{\flat}T_1=1$ and (d) $T_1$ is doubled up. 
This is sufficient because it  tells us that  the equivalence classes of  the power spectrum are related via symplectic similarity transformations (and they are the same  gauge transformations  as those obtained  from the transfer function \cite{levitt}). 
The outline of how we show (a)-(d) is given in the following  three steps. The complete proof can be found in  \ref{p1}.

\begin{enumerate}
\item[1)] Firstly, using the pattern in the $\tilde{A}$, $\tilde{B}$ and $\tilde{C}$ matrices defined above, we show that the following holds:
$$\mathcal{O}=\mathcal{O}\left(\begin{smallmatrix} T_4^{\flat} &0\\ -T_3^{\flat}& T_1^{\flat}\end{smallmatrix}\right)T.$$ 
And so  because $\mathcal{O}$ has full rank, we have
$$\left(\begin{smallmatrix} T_4^{\flat} &0\\ -T_3^{\flat} &T_1^{\flat}\end{smallmatrix}\right)T=1.$$
\item[2)] We will then 
 show that:
$$\mathcal{O}\left(\begin{smallmatrix} T_4^{\flat}-T_1^{\flat}\\-T_3^{\flat}\end{smallmatrix}\right)=0.$$
This implies that $T_3=0$ and $T_1=T_4$.
\item[3)] Combing Steps 1) and 2) it is clear that $T$ must be of the form  $$T=\left(\begin{smallmatrix}T_1&0\\0&T_1\end{smallmatrix}\right)$$ with $T_1^{\flat}T_1=1.$
Finally we show that $T_1$ is doubled-up. 
\end{enumerate}
\end{proof}



\subsection{Identification method}\label{meth}

Suppose that we have constructed the power spectrum  from the input-output data, for instance by treating it as a transfer function and using one of the techniques of \cite{Ljung}. Here we outline a method to construct a globally minimal system realisation from the power spectrum.
 The realisation is obtained indirectly by first finding a non-physical realisation and then constructing a physical one from this by applying a criterion developed in \cite{zhou}. The construction is  similar to the one   used in \cite{levitt} for the transfer function realisation problem.

We have already seen many times that the power spectrum may be treated  as if it were a transfer function. Therefore, let $\left(\tilde{A}_0, \tilde{B}_0, \tilde{C}_0, V_{\mathrm{vac}}\right)$ constitute a minimal realisation of $\Psi(s)$, i.e., 
$$\Psi(s)J=V_{\mathrm{vac}}+ \tilde{C}_0\left(s-\tilde{A}_0\right)^{-1}\tilde{B}_0.$$
Further, let us  assume that $\tilde{A}_0, \tilde{B}_0, \tilde{C}_0$ are of the form 
$$\tilde{A}_0=\left(\begin{smallmatrix} -{A}^{\flat}_0&0\\0&{A}_0\end{smallmatrix}\right)\quad
\tilde{B}_0=\left(\begin{smallmatrix}B_1\\B_2\end{smallmatrix}\right) \quad
\tilde{C}_0=\left(\begin{smallmatrix}C_1\\C_2\end{smallmatrix}\right),$$
with, $A_0$, $B_1$ and $C_2$ doubled up and $A_0$ is stable. For example, in \ref{nog} such a realisation is found for an $n$-mode globally minimal system, with matrices $(A,C)$, possessing $2n$ distinct poles each with non-zero imaginary part

Now, by minimality, any other realisation of the transfer function can be generated by the similarity transformation 
\begin{equation}\label{create}\tilde{A}=T\tilde{A}_0T^{-1}\quad
\tilde{B}=T\tilde{B}_0\quad
\tilde{C}=\tilde{C}_0T^{-1}.\end{equation}
The problem here is that in general these matrices may not describe a genuine quantum system in the sense that from a given $\left(\tilde{A}, \tilde{B}, \tilde{C}\right)$ one cannot reconstruct the pair $(\Omega, C)$ describing the power spectrum. Our goal is to find a special transformation $T$ mapping $(\tilde{A}_0 , \tilde{B}_0 , \tilde{C}_0)$ to a triple $(\tilde{A}, \tilde{B}, \tilde{C})$ that is physical.

Firstly, as $\tilde{A}_0$ and the physical $\tilde{A}$ we seek are both proper LBT, then by Lemma 
\ref{hud} we may restrict  $T$ to be of the  form
$$T=\left(\begin{smallmatrix}T_1&0\\T_2&T_3\end{smallmatrix}\right), \quad T^{-1}=\left(\begin{smallmatrix}T^{-1}_1&0\\-T_3^{-1}T_2T_1^{-1}&T_3^{-1}\end{smallmatrix}\right).$$
 Using this together with \eqref{create} and    $(\tilde{A}, \tilde{B}, \tilde{C})$ in \eqref{cask} gives:
\begin{align}
\label{one}A^{\flat}=T_1A_0^{\flat}T_1^{-1}\\
\label{two}C^{\flat}V_{\mathrm{vac}}C=-T_2{A}^{\flat}_0T_1^{-1}-T_3A_0T_3^{-1}T_2T_1^{-1}\\
\label{three}A=T_3A_0T_3^{-1}\\
\label{four}-C^{\flat}=T_1B_1\\
\label{five}-C^{\flat}V_{\mathrm{vac}}=T_2B_1+T_3B_2\\
\label{six}-V_{\mathrm{vac}}C=C_1T_1^{-1}-C_2T_3^{-1}T_2T_1^{-1}\\
\label{seven}C=C_2T_3^{-1}.
\end{align}
For $(A, C)$ to correspond to a quantum system it must satisfy the physical realisability conditions:
$A+A^{\flat}+C^{\flat}C=0$  \cite{Indep}.
Applying this condition to \eqref{three} and \eqref{seven} and then again to \eqref{one} and \eqref{four} leads to the following equations:
\begin{align}
\label{eight}\left(T_3^{\flat}T_3\right)A_0+A_0^{\flat}\left(T_3^{\flat}T_3\right)+C_2^{\flat}C_2=0\\
\label{nine}
A_0^{\flat}\left(T_1^{\flat}T_1\right)^{-1}+\left(T_1^{\flat}T_1\right)^{-1}A_0+B_1B_1^{\flat}=0.
\end{align}
Next, as quantum system is  stable  $A_0$ must be  Hurwitz (because it is similar to $A$), therefore  \eqref{eight} and \eqref{nine} have unique solutions 
$\left(T_3^{\flat}T_3\right)$ and $\left(T_1^{\flat}T_1\right)^{-1}$ respectively (see   \cite{levitt} for the explicit form of these). Moreover, these solutions will necessarily be of doubled-up form due to the fact $A_0, B_1$ and $C_2$ were.
Therefore, using Lemma 1 in \cite{levitt} we can find doubled-up $T_1$ and $T_3$ from these uniquely (up to the non-identifiable symplectic equivalence class in Theorem \ref{main}). 

The upshot of these results is that we may ultimately write down a realisation of the system $(A, C)$ using \eqref{eight}  or alternatively from \eqref{nine}. By Theorem \ref{main} both solutions are guaranteed to coincide (bar any unidentifiable symplectic matrix) and give a unique (up to such a symplectic transformation) realisation of the power spectrum, hence we are done.

For completeness we may write down the unique solution $T_2$ given the solutions $T_1$ and $T_3$ so to obtain the full realisation \eqref{cask} of the power spectrum. To this end, 
suppose that the solutions $T_1$ and $T_2$ from \eqref{eight} and \eqref{nine} lead to (physical) realisations $(A, C)$ and $(\hat{A}, \hat{C})$ that differ by an (unidentifiable) symplectic. That is, $A=S\hat{A}S^{\flat}$ and $C=\hat{C}S^{\flat}$. Then from \eqref{two} we have
$$S\hat{C}^{\flat}V_{\mathrm{vac}}\hat{C}+\left(T_2T_1^{-1}S\right)\hat{A}^{\flat}+\hat{A}\left(T_2T_1^{-1}S\right)=0,$$
which has been obtained by substituting \eqref{one} and \eqref{three} into  \eqref{two}.
 This solution $\left(T_2T_1^{-1}S\right)$ can be found uniquely, hence $T_2$ can be found uniquely from this. Note that $T_2$ will not be of doubled-up type, which is to be expected.

\section{Outlook}

Our main result is that under  global minimal and pure stationary inputs  the power spectrum contains as much information as the transfer function, i.e., their classes of equivalent systems are the same in both functions.  Therefore, no information is lost by utilising stationary inputs rather than time-dependent inputs. As a corollary to these results it is not too difficult to prove a similar statement for the subset of passive systems. 

It would be interesting to understand whether these results would hold for mixed input states.  However, clearly   the equivalence between global minimality and mixedness of the stationary state from \cite{levitt} will not hold. Therefore, understanding whether or not a system is globally minimal for a mixed input requires further theory. Also, we could ask the same identifiability questions for the case of unknown inputs or allow for  static scattering or squeezing in the field. We intend to address these extensions in future works. 

Given that we now understand what is identifiable, the next step is to understand how well parameters can be estimated. In the time-dependent approach this has been done for passive systems in \cite{Guta} but no such work exists for active systems or in the stationary approach at all. As a side note, it should be possible to find the gauge transformations in the power spectrum that we found here as the  directions in phase space along which the  \textit{quantum Fisher information}\footnote{Recall that the Q.F.I gives a measure of the optimal estimation precision using the best measurement and estimator.} vanishes. 
Lastly, it would be interesting to consider these identifiability problems in the more realistic scenario of noisy QLSs. In a QLS noise may be modelled by the inclusion of additional channels that cannot be monitored. Understanding what can be identified here will likely be more challenging.

\appendix

\section{Proof of Lemma \ref{hud}}\label{goh}

 \begin{proof}
Firstly define $\{e_1, ..., e_{4n}\}$ as the canonical basis of $\mathbb{C}^{4n}$. 
%
%
By property \eqref{pil1} of  proper LBT matrices it is clear that $y^{(i)}:= \left(\begin{smallmatrix}0\\y_2^{(i)}\end{smallmatrix}\right)\in\mathrm{Span}\{e_{2n+1}, ..., e_{4n}\}$. Further, as there are $2n$ of them they must form a basis of $\mathrm{Span}\{e_{2n+1}, ..., e_{4n}\}$.
Suppose $y^{(i)}$ has generalised eigenvector rank $m_i$, then as
 as $\tilde{A}'=T\tilde{A}T^{-1}$ we have 
%
%
\begin{align*}
\left(\tilde{A}'-\lambda^{(i)}\right)^{m_i}  Ty^{(i)}&=\left(T\tilde{A}T^{-1}-\lambda^{(i)}\right)^{m_i}  Ty^{(i)}\\
&=T\left(\tilde{A}-\lambda^{(i)}\right)^{m_i}  y^{(i)}\\
&=0.
\end{align*}
Therefore, $Ty^{(i)}$ are generalised eigenvectors of $\tilde{A}'$ associated to $\lambda^{(i)}$. Hence, because $\tilde{A}'$ is also assumed to be proper LBT, it follows that $\mathrm{Span}\{Ty^{(i)}\}\subset\mathrm{Span}\{e_{2n+1}, ..., e_{4n}\}$. Finally, 
\begin{align*}
T\mathrm{Span}\{e_{2n+1}, ..., e_{4n}\}&=T\mathrm{Span}\{y^{(i)}\}\\
&=\mathrm{Span}\{Ty^{(i)}\}\\
&\subset\mathrm{Span}\{e_{2n+1}, ..., e_{4n}\}.
\end{align*}
The invertibility of $T$ has been used in getting from the first to the second line. 
This implies that $T$ is LBT, as required. 
\end{proof}

\section{Proof of Theorem \ref{main}}\label{p1}

As outlined in the proof sketch,  we need to show 1-3.

\subsection{Step 1:}\label{s1}
Firstly, the condition $\tilde{B}'=T\tilde{B}$ is equivalent to  
$$\left(\begin{smallmatrix}-C'^{\dag}V_{\mathrm{vac}}\\{C'}^{\dag}\end{smallmatrix}\right)=K\Sigma T\Sigma K \left(\begin{smallmatrix}-{C}^{\dag}V_{\mathrm{vac}}\\{C}^{\dag}\end{smallmatrix}\right),$$
where
$$K=\left(\begin{smallmatrix} J&0\\0&-J\end{smallmatrix}\right) \quad \mathrm{and} \quad \Sigma=\left(\begin{smallmatrix} 0&1\\1&0\end{smallmatrix}\right).$$ 
Hence 
$$\left(\begin{smallmatrix}-V_{\mathrm{vac}}\tilde{C'}&\tilde{C'}\end{smallmatrix}\right)=\left(\begin{smallmatrix}-V_{\mathrm{vac}}\tilde{C}&\tilde{C}\end{smallmatrix}\right)K\Sigma T^{\dag}\Sigma K.$$ 
Therefore, combining this the condition $\tilde{C}'=\tilde{C}T^{-1}$ we have
\begin{equation}\label{in10}
\tilde{C}=\tilde{C}K\Sigma T^{\dag}\Sigma KT.
\end{equation}
Now, 
\begin{align}
\nonumber \tilde{A}'&=-K\Sigma\tilde{A}'^{\dag}\Sigma K\\
\nonumber&=-K\Sigma \left(T^{\dag}\right)^{-1} \tilde{A}^{\dag}T^{\dag}\Sigma K\\
\label{as}&=K\Sigma  \left(T^{\dag}\right)^{-1} \Sigma K \tilde{A}K\Sigma T^{\dag}\Sigma K,
\end{align}
where $\tilde{A}'=T\tilde{A}T^{-1}$ has been used to obtain the second line.

%
%
And so 
\begin{align*}
\tilde{C} \tilde{A}T^{-1}&=\tilde{C}' \tilde{A}'\\
&=\tilde{C}'K\Sigma  \left(T^{\dag}\right)^{-1} \Sigma K \tilde{A}K\Sigma T^{\dag}\Sigma K\\
&=\left(\tilde{C}T^{-1}K\Sigma  \left(T^{\dag}\right)^{-1} \Sigma K\right) \tilde{A}K\Sigma T^{\dag}\Sigma K\\
&=\tilde{C} \tilde{A}K\Sigma T^{\dag}\Sigma K,
\end{align*}
where  \eqref{in10} has been used to obtain the fourth line.
Thus
\begin{equation}\label{in2}
\tilde{C} \tilde{A}=\left(\tilde{C} \tilde{A}\right)K\Sigma T^{\dag}\Sigma KT.
\end{equation}

\begin{Claim}
%
\begin{equation}
\tilde{C} \tilde{A}^k=\left(\tilde{C} \tilde{A}^k\right)K\Sigma T^{\dag}\Sigma KT.
\end{equation}
for all $k\geq0$. 
\end{Claim}

\begin{proof}
We  prove this by induction. Note that we already know it to be true for $k=0$ and $k=1$ (see \eqref{in10} and \eqref{in2}). To this end, suppose that it is true for $k-1$. Therefore,
\begin{align*}
\tilde{C}'\tilde{A}'^k&=\tilde{C}'\left(\tilde{A}'\right)^{k-1}\tilde{A}'\\
&=\tilde{C}'\left(\tilde{A}'\right)^{k-1}K\Sigma  \left(T^{\dag}\right)^{-1} \Sigma K \tilde{A}K\Sigma T^{\dag}\Sigma K\\
&=\left(\tilde{C}\tilde{A}^{k-1}T^{-1}K\Sigma  \left(T^{\dag}\right)^{-1} \Sigma K \right)\tilde{A}K\Sigma T^{\dag}\Sigma K\\
&=\tilde{C}\tilde{A}^{k}K\Sigma T^{\dag}\Sigma K.
\end{align*}
by using a combination of \eqref{as} and \eqref{c1}. Finally, using the observation $\tilde{C}'\tilde{A}'^k=\tilde{C}\tilde{A}^kT^{-1}$ completes the proof. %
\end{proof}
Finally, following this claim we  have:
\begin{align*}
\mathcal{O}&=\mathcal{O}K\Sigma T^{\dag}\Sigma KT
\\&=\mathcal{O}\left(\begin{smallmatrix} T_4^{\flat} &0\\ -T_3^{\flat}& T_1^{\flat}\end{smallmatrix}\right)T.
\end{align*}

 %
%

\subsection{Step 2:}
For this step it is sufficient to prove the following claim. 

\begin{Claim}
$$\tilde{C}\tilde{A}^k\left(\begin{smallmatrix}T_4^{\flat}-T_1^{\flat}\\-T_3^{\flat}\end{smallmatrix}\right)=0$$
for all $k=0,1,2,...$.
\end{Claim}

\begin{proof}
Using the results of \ref{s1} we know that equivalent systems   are related via
\begin{equation}\label{start}
\tilde{C'}\tilde{A'}^k=\tilde{C}\tilde{A}^k\left(\begin{smallmatrix} T_4^{\flat}&0\\-T^{\flat}_3&T_1^{\flat}\end{smallmatrix}\right).
\end{equation} 
Also     note that the condition $C'A'^k=CA^kT^b_1$ holds.

We first see this result for $k=0$. Equation \eqref{start} for $k=0$ reads 
$$\left(\begin{smallmatrix}-V_{\mathrm{vac}}C', C'\end{smallmatrix}\right)=\left(\begin{smallmatrix}-V_{\mathrm{vac}}CT_4^{\flat}-CT_3^{\flat}, CT_1^{\flat}\end{smallmatrix}\right).$$ 
Therefore, adding the first entry to $V_{\mathrm{vac}}$ times the second entry:
$$0=-V_{\mathrm{vac}}C\left(T_4^{\flat}-T_1^{\flat}\right)+C\left(-T_3^{\flat}\right),$$
which shows the result for $k=0$.

The result for $k\in\mathbb{N}$ goes along the same lines, but is a little more involved. Firstly, observe that $\tilde{A}^{k}$ may be written as
$$\tilde{A}^k=\left(\begin{smallmatrix}  \left(-A^{\flat}\right)^k &0\\ e_k &A^k\end{smallmatrix}\right),$$
where $e_k=A^0C^{\flat}V_{\mathrm{vac}}C\left(-A^{\flat}\right)^{k-1}+...+A^{k-1}C^{\flat}V_{\mathrm{vac}}C\left(-A^{\flat}\right)^{0}$ (and similarly for the primed matrices).
Now, from \eqref{start} we have 
\begin{align}
\nonumber&\left(\begin{smallmatrix} -V_{\mathrm{vac}}C'\left(-A'^{\flat}\right)^k+C'A'^{k-1}C'^{\flat}V_{\mathrm{vac}}C'-C'e'_{k-1}A'^{\flat},&C'A'^k\end{smallmatrix}\right)\\
&\label{combo}=\left(-\begin{smallmatrix} V_{\mathrm{vac}}C\left(-A^{\flat}\right)^kT_4^b+CA^{k-1}C^{\flat}V_{\mathrm{vac}}CT_4^{\flat}-Ce_{k-1}A^{\flat}T_4^{\flat}-CA^kT_3^{\flat},&CA^kT_1^{\flat}\end{smallmatrix}\right). 
\end{align}
Again adding the first block to $V_{\mathrm{vac}}$ times the second block gives 
\begin{equation}\label{pod}
H'=\tilde{C}\tilde{A}^k\left(\begin{smallmatrix} T_4^{\flat}-T_1^{\flat}\\
-T^{\flat}_3 \end{smallmatrix}\right)+HT_1^b,
\end{equation}
where
\begin{align*}
H:&=-V_{\mathrm{vac}}C\left(-A^{\flat}\right)^k+CA^{k-1}C^{\flat}V_{\mathrm{vac}}C-Ce_{k-1}A^{\flat}+V_{\mathrm{vac}}CA^k\\
H':&=-V_{\mathrm{vac}}C'\left(-A'^{\flat}\right)^k+C'A'^{k-1}C'^{\flat}V_{\mathrm{vac}}C'-C'e'_{k-1}A'^{\flat}+V_{\mathrm{vac}}C'A'^k
\end{align*}

Now, observe that
\begin{align}
\nonumber H'&=V_{\mathrm{vac}}C'A'^k+\left(V_{\mathrm{vac}}C'\left(-A'^{\flat}\right)^{k-1}-C'e'_{k-1}\right)A'^{\flat}+C'A'^{k-1}C'^{\flat}V_{\mathrm{vac}}C'\\
\nonumber&=V_{\mathrm{vac}}C'A'^k+\left(V_{\mathrm{vac}}C'\left(-A'^{\flat}\right)^{k-1}-C'e'_{k-1}\right)\left(-A'-C'^{\flat}C'\right)\\
\nonumber&+C'A'^{k-1}C'^{\flat}V_{\mathrm{vac}}C'\\
\nonumber&=V_{\mathrm{vac}}C'A'^k+\left(-V_{\mathrm{vac}}C'\left(-A'^{\flat}\right)^{k-1}+C'e'_{k-1}   \right)A'
\\\nonumber& +  \left(-V_{\mathrm{vac}}C'\left(-A'^{\flat}\right)^{k-1}C'^{\flat}+C'e'_{k-1}C'^{\flat}+ C'A'^{k-1}C'^{\flat}V_{\mathrm{vac}}    \right)C'    \\
&\label{roff}=V_{\mathrm{vac}}C'A'^k+G'_{k-1}A'-\tilde{C}'\left(\tilde{A}'\right)^{k-1}\tilde{B}'C',
\end{align}
where $G'_k:=-V_{\mathrm{vac}}C'\left(-A'^{\flat}\right)^{k}+C'e'_{k}$. Here we have used the realisability condition $A+A^{\flat}+C^{\flat}C=0$ on the second line and then rearranged. 

Now, let us obtain a recursive expression for $G_k$. Firstly, using the definition of $e_k$ and the substitution $A'+A'^{\flat}+C'^{\flat}C'=0$:

\begin{align*} 
G'_k&=-V_{\mathrm{vac}}C'\left(-A'^{\flat}\right)^{k}   +    \sum^{k-1}_{j=0} C'A'^{k-1-j}C'^{\flat}V_{\mathrm{vac}}C'\left(-A'^{\flat}\right)^{j} 
\\
&=-V_{\mathrm{vac}}C'\left(-A'^{\flat}\right)^{k-1}\left(A'+C'^{\flat}C'\right)
\\&+  \sum^{k-1}_{j=1}  C'A'^{k-1-j}C'^{\flat}V_{\mathrm{vac}}C'\left(-A'^{\flat}\right)^{j-1} \left(A'+C'^{\flat}C'\right)   \\
&            +
C'A'^{k-1}C'^{\flat}V_{\mathrm{vac}}C'\left(-A'^{\flat}\right)^{0}
\end{align*}
Rearranging this and using the definition of $e_k$ again we obtain 
\begin{align*}
G'_k&=\left(   -V_{\mathrm{vac}}C'\left(-A'^{\flat}\right)^{k-1}C'^{\flat}+ C'A'^{k-1}C'^{\flat}V_{\mathrm{vac}}\right.\\
&\left.+C'\left[ \sum^{k-2}_{j=0}   A'^{k-2-j}C'^{\flat}V_{\mathrm{vac}}C'\left(-A'^{\flat}\right)^{j}
   \right]C'^{\flat}      \right)C'\\
&+\left(   -V_{\mathrm{vac}}C'\left(-A'^{\flat}\right)^{k-1}  +C'\sum^{k-2}_{j=0}  A'^{k-2-j}C'^{\flat}V_{\mathrm{vac}}C'\left(-A'^{\flat}\right)^{j}\right)A'\\
&
=\left( -V_{\mathrm{vac}}C'\left(-A'^{\flat}\right)^{k-1}C'^{\flat}+ C'A'^{k-1}C'^{\flat}V_{\mathrm{vac}}+C'e'_{k-1}C'^{\flat}\right)C'\\
&+\left( -V_{\mathrm{vac}}C'\left(-A'^{\flat}\right)^{k-1} +C'e'_{k-1}\right)A'\\
&=-\tilde{C}'\tilde{A}'^{k-1}\tilde{B}'C'+G_{k-1}A'.
\end{align*}
%
%
%
Also note that 
\begin{align*}
G_1&=V_{\mathrm{vac}}C'A'^{\flat}+C'C'^{\flat}V_{\mathrm{vac}}C'\\
&=-V_{\mathrm{vac}}C'A'+\left(-V_{\mathrm{vac}}C'C'^{\flat}+C'C'^{\flat}V_{\mathrm{vac}}\right)C'\\
&=-V_{\mathrm{vac}}C'A'-\tilde{C}'\tilde{B}'C'.
\end{align*}

Using our recursive expression for $G_k$, and continuing on from \eqref{roff} we have 
\begin{align}\label{nun}
\nonumber H'=&V_{\mathrm{vac}}C'A'^k-\tilde{C}'\tilde{A}'^{k-1}\tilde{B}'C'+G'_{k-1}A'\\
&=V_{\mathrm{vac}}C'A'^k-\tilde{C}'\tilde{A}'^{k-1}\tilde{B}'C'-\tilde{C}'\tilde{A}'^{k-2}\tilde{B}'C'A'+G_{k-2}\tilde{A}'^2\nonumber\\
&\,\,\vdots \quad \quad\quad \quad\quad   \vdots \quad\quad\quad\quad\quad \vdots\quad \quad\quad \quad\quad   \vdots\nonumber\\
&=V_{\mathrm{vac}}C'A'^k-\sum^{k-1}_{j=1}\tilde{C}'\tilde{A}'^{j}\tilde{B}'C'A'^{k-1-j}+G_1A'^{k-1}\nonumber\\
&=V_{\mathrm{vac}}C'A'^k-  \sum^{k-1}_{j=0}\tilde{C}'\tilde{A}'^{j}\tilde{B}'C'A'^{k-1-j}                -V_{\mathrm{vac}}C'A'^k\nonumber\\
&=-\sum^{k-1}_{j=0}\tilde{C}'\tilde{A}'^{j}\tilde{B}'C'A'^{k-1-j}         .
\end{align}
%
%
%
%
%
Furthermore, as $\tilde{C}'\tilde{A}'^{k-1}\tilde{B}'=\tilde{C}\tilde{A}^{k-1}\tilde{B}$ for all $k$ and $C'A'^k=CA^kT^b_1$,  then we may conclude that 
\begin{align}\label{compare1}
H'=-  \left(   \sum^{k-1}_{j=0}\tilde{C}'\tilde{A}'^{j}\tilde{B}'C'A'^{k-1-j}                    \right)T_1^{\flat}.
\end{align}

On the other hand, by using an identical argument to above, 
\begin{equation}\label{compare2}
H=-\sum^{k-1}_{j=0}\tilde{C}\tilde{A}^{j}\tilde{B}CA^{k-1-j}.      
\end{equation}  
Therefore, using \eqref{compare1} and \eqref{compare2} in \eqref{pod} completes the proof. 
\end{proof}

\subsection{Step 3} To show that the system is doubled-up we use the observability of the quantum system. Observe that $C_1A_1^k$, $C_2A_2^k$ must be of the of this doubled up form for $k\in\{0,1,2,...\}$. Writing $C_1A_1^k$, $C_2A_2^k$ and $T_1$ as $\left(\begin{smallmatrix} P_{(k)} &Q_{(k)}\\\overline{Q}_{(k)}&\overline{P}_{(k)}\end{smallmatrix}\right)$, $\left(\begin{smallmatrix} P'_{(k)} &Q'_{(k)}\\\overline{Q}_{(k)}'&\overline{P}_{(k)}'\end{smallmatrix}\right)$ and $T_1=  \left(\begin{smallmatrix} S_1 &S_2\\S_3&S_4\end{smallmatrix}\right)$, and using the result, $C_1A_1^k= C_2A_2^k T_1^\flat $, 
it follows that 
\[P_{(k)}(S_1^{\dag}-S_4^T)+Q_{(k)}(S_3^T-S_2^{\dag})=0\]
\[\overline{Q}_{(k)}(S_1^{\dag}-S_4^T)+\overline{P}_{(k)}(S_3^T-S_2^{\dag})=0.\]
Hence 
\[\mathcal{O}\left[\begin{smallmatrix}    S^{\dag}_1-S_4^T\\  S_3^T-S_2^{\dag}\end{smallmatrix}\right]=0\] and by using the fact that $\mathcal{O}$ is full rank implies that 
\[T_1=\left(\begin{smallmatrix} S_1 &S_2\\\overline{S}_2&\overline{S}_1\end{smallmatrix}\right).\] 

\section{Finding a classical realisation of the power spectrum for Section \ref{meth}}\label{nog}

We assume that the matrix  $A$ for   the $n$-mode minimal system, $(A, C)$, possesses  $2n$ distinct eigenvalues each with non-zero imaginary part. This requirement  can be seen to be generic in the space of all quantum systems \cite{cascade}. 

Firstly, observe that if  $\lambda_i$ is a complex eigenvalue of $A$ with right eigenvector $\left(\begin{smallmatrix}R_i\\S_i\end{smallmatrix}\right)$ and left eigenvector  $\left(U_i, V_i\right)$, then $\overline{\lambda}_i$ also an eigenvalue with right eigenvector $\left(\begin{smallmatrix}\overline{S}_i\\\overline{R}_i\end{smallmatrix}\right)=\Sigma\overline{\left(\begin{smallmatrix}R_i\\S_i\end{smallmatrix}\right)}$ and left eigenvector $   \left(\overline{V}_i, \overline{U}_i\right)= \overline{\left(U_i, V_i\right)}\Sigma_n$, where     $R_i, S_i\in\mathbb{C}^{1\times n}$, $U_i, V_i\in\mathbb{C}^{n\times1}$ and $\Sigma_n:=\left(\begin{smallmatrix} 0_n&1_n\\1_n&0_n\end{smallmatrix}\right)$. 
This property follows from the fact that $A$ has the doubled-up form  $A:=\Delta\left(A_{-}, A_{+}\right)$.
Furthermore, from the  system  \eqref{cask} $\tilde{A}$ may be diagonalised as $\tilde{A}=P\tilde{A}_0P^{-1}$ where 
$$\tilde{A}_0=\left(\begin{smallmatrix}-{A}^{\flat}_0&0\\0&A_0\end{smallmatrix}\right)$$
and $A_0$ is diagonal and doubled-up. Here $P$ and $P^{-1}$ are lower block triangular (Lemma \ref{hud})
written as 
$$P=\left(\begin{smallmatrix} P_1&0\\P_2&P_3\end{smallmatrix}\right) \quad \mathrm{and}\quad
P^{-1}=\left(\begin{smallmatrix} P^{-1}_1&0\\-P_3^{-1}P_2P_1^{-1}&P^{-1}_3\end{smallmatrix}\right),$$
where 
$$P_3=\left(\begin{smallmatrix}R_1&\hdots &R_n&S_1&\hdots&S_n\\ \overline{S}_1&\hdots&\overline{S}_n&\overline{R}_1&\hdots&\overline{R}_n\end{smallmatrix}\right) \quad\mathrm{and}\quad
P^{-1}_1=\left(\begin{smallmatrix}U_1&V_1\\\vdots&\vdots\\
U_n&V_n\\
 \overline{V}_1&\overline{U}_1\\
 \vdots&\vdots\\
 \overline{V}_n&\overline{U}_n\end{smallmatrix}\right).
$$
Hence, the power spectrum, $\Psi(s)J$, of \eqref{cask} may be written 
\begin{align}
\label{gut}
V_{\mathrm{vac}}-\left(-V_{\mathrm{vac}}CP_1+CP_2, CP_3\right)\left(\begin{smallmatrix}s+A^{\flat}_0&0\\0 &s-A_0\end{smallmatrix}\right)\left(\begin{smallmatrix}P_1^{-1}C^{\flat}\\-P_3^{-1}P_2P_1^{-1}C^{\flat}+ P_3^{-1}C^{\flat}V_{\mathrm{vac}}\end{smallmatrix}\right).
\end{align}

We can construct a minimal realisation called \textit{Gilbert's realisation} \cite{zhou} by expanding as partial fractions:
\begin{equation}\label{gut1}
\Psi(s)J=V_{\mathrm{vac}}+\sum_{i=1}^n \frac{I_i}{(s+\overline{\lambda}_i)} +\frac{K_i}{(s+\lambda_i)}
+\frac{T_i}{(s-\lambda_i)}+\frac{W_i}{(s-\overline{\lambda}_i)},
\end{equation}
with $\mathrm{Re}(\lambda_i)<0$.
The matrices $I_i, K_i, T_i, W_i$ are necessarily rank-one. Therefore there exist matrices $B_{1,i}, B_{2,i}, B'_{1,i}, B'_{2,i}\in\mathbb{C}^{1\times 2m}$ and $C_{1,i}, C_{2,i}, C'_{1,i}, C'_{2,i}\in\mathbb{C}^{2m\times 1}$  such that 
\begin{equation*}
C_{1,i}B_{1,i}=I_i, C'_{1,i}B'_{1,i}=K_i\quad \,\,\,\mathrm{and}\,\,\, C_{2,i}B_{2,i}=T_i, \quad C'_{2,i}B'_{2,i}=W_i
\end{equation*}
and are each uniquely determined from $I_i, K_i, T_i, W_i$ up to a constant\footnote{For example $\frac{1}{\nu}C_{1,i}$ and $\nu B_{1,i}$ are also solutions to $I_i$, where $\nu$ is a constant.}.
The Gilbert realisation $\tilde{A}_0, \tilde{B}_0, \tilde{C}_0$ is 
\begin{equation*}
\tilde{A}_0:=\mathrm{diag}\left(-\overline{\lambda}_1, ..., -\overline{\lambda}_n, -\lambda_1, ..., -\lambda_n, \lambda_1, ..., \lambda_n,\overline{\lambda}_1,..., \overline{\lambda}_n\right),
\end{equation*}
\begin{equation*}
\tilde{B}_0:=\left[\begin{smallmatrix}B_1\\B_2\end{smallmatrix}\right], \quad 
\tilde{C}_0:= 
\left[C_1, C_2\right]
\end{equation*}
where 
\begin{align*}
& B_1:=\left[\begin{smallmatrix}B_{1,1}\\\vdots\\B_{1,n}\\B'_{1,1}\\\vdots\\B'_{1,n}\end{smallmatrix}\right] 
\quad 
B_2:=\left[\begin{smallmatrix}B_{2,1}\\\vdots\\B_{2,n}\\B'_{1,1}\\\vdots\\B'_{1,n}\end{smallmatrix}\right],   \\
 &C_1:=\left[\begin{smallmatrix}C_{1,1}&\hdots&C_{1,n}&C'_{1,1}&\hdots&C'_{1,n}\end{smallmatrix}\right], \\& C_2:=\left[\begin{smallmatrix}C_{2,1}&\hdots&C_{2,n}&C'_{2,1}&\hdots&C'_{2,n}\end{smallmatrix}\right] .
\end{align*}
At the moment this Gilbert realisation doesn't satisfy the properties required by Section \ref{meth}, i.e., $B_1$ and $C_2$ are not doubled-up. We can take care of this in the following way. 
Firstly, in this realisation  $I_i$ is equal to the $i^{\mathrm{th}}$  column of $\left(-V_{\mathrm{vac}}CP_1+CP_2\right)$ multiplied by the $i^{\mathrm{th}}$  row of $P_1^{-1}C^{\flat}$ and $K_i$ is equal to the $(n+i)^{\mathrm{th}}$ column of $\left(-V_{\mathrm{vac}}CP_1+CP_2\right)$ multiplied by the $(n+i)^{\mathrm{th}}$ row of $P_1^{-1}C^{\flat}$ (see \eqref{gut}).  Therefore,  the $i^{\mathrm{th}}$ row of $B_1$ differs from the $i^{\mathrm{th}}$ row of the doubled-up matrix $P_1^{-1}C^{\flat}$ by an (unknown) multiplicative constant.  
Finally, by   multiplying the rows of $B_1$ in our Gilbert realisation by suitable constants (and hence multiplying the corresponding columns of $C_1$ by the inverse of these constants so that the power spectrum remains unchanged) we can obtain a doubled-up $B_1$. A similar technique may be used to obtain a doubled-up $C_2$ by using the fact that $CP_3$ is doubled-up.



\bibliographystyle{elsarticle-harv}
\bibliography{references.bib} 

\begin{thebibliography}{40}
\expandafter\ifx\csname natexlab\endcsname\relax\def\natexlab#1{#1}\fi
\expandafter\ifx\csname url\endcsname\relax
  \def\url#1{\texttt{#1}}\fi
\expandafter\ifx\csname urlprefix\endcsname\relax\def\urlprefix{URL }\fi

\bibitem[{Anderson et~al.(1966)Anderson, Newcomb, Kalman, and Youla}]{anders}
Anderson, B., Newcomb, R., Kalman, R., Youla, D., 1966. Equivalence of linear
  time-invariant dynamical systems. J.Franklin Inst. 281~(5), 371--378.

\bibitem[{Astrom and Wittenmark(2008)}]{Astrom}
Astrom, K.~J., Wittenmark, B., 2008. Adaptive control. Dover Publications.

\bibitem[{Bouten(2004)}]{LB2}
Bouten, L., 2004. Filtering and control in quantum optics. arXiv preprint:
  0410080.

\bibitem[{Bouten et~al.(2007)Bouten, Van~Handel, and James}]{LB}
Bouten, L., Van~Handel, R., James, M.~R., 2007. An introduction to quantum
  filtering. SIAM Journal on Control and Optimization 46~(6), 2199--2241.

\bibitem[{Davis(1963)}]{Davies}
Davis, M., 1963. Factoring the spectral matrix. IEEE Trans. Auto. Control
  8~(4), 296--305.

\bibitem[{Doherty and Jacobs(1999)}]{DOHERTY}
Doherty, A.~C., Jacobs, K., 1999. Feedback control of quantum systems using
  continuous state estimation. PRA 60~(4), 2700.

\bibitem[{Dong and Petersen(2010)}]{Dong}
Dong, D., Petersen, I.~R., 2010. Quantum control theory and applications: a
  survey. IET Control Theory \& Applications 4~(12), 2651--2671.

\bibitem[{Gardiner and Zoller(2004)}]{GZ}
Gardiner, C., Zoller, P., 2004. Quantum noise: a handbook of Markovian and
  non-Markovian quantum stochastic methods with applications to quantum optics.
  Vol.~56. Springer Science \& Business Media.

\bibitem[{Glover and Willems(1974)}]{glover}
Glover, K., Willems, J., 1974. Parametrizations of linear dynamical systems:
  canonical forms and identifiability. IEEE Trans. Auto. Control 19~(6),
  640--646.

\bibitem[{Gough et~al.(2010)Gough, James, and Nurdin}]{squeezing}
Gough, J.~E., James, M., Nurdin, H., 2010. Squeezing components in linear
  quantum feedback networks. PRA 81~(2), 023804.

\bibitem[{Gough and Zhang(2015)}]{Indep}
Gough, J.~E., Zhang, G., 2015. On realization theory of quantum linear systems.
  Automatica 59, 139--151.

\bibitem[{Gu{\c{t}}{\u{a}} and Kiukas(2015)}]{GUTA3}
Gu{\c{t}}{\u{a}}, M., Kiukas, J., 2015. Equivalence classes and local
  asymptotic normality in system identification for quantum markov chains.
  Comm. Math. Phys 335~(3), 1397--1428.

\bibitem[{Gu{\c{t}}{\u{a}} and Kiukas(2016)}]{GUTA4}
Gu{\c{t}}{\u{a}}, M., Kiukas, J., 2016. Information geometry and local
  asymptotic normality for multi-parameter estimation of quantum markov
  dynamics. arXiv preprint: 1601.04355.

\bibitem[{Gu{\c{t}}{\u{a}} and Yamamoto(2013)}]{Guta}
Gu{\c{t}}{\u{a}}, M., Yamamoto, N., 2013. Systems identification for passive
  linear quantum systems: the transfer function approach. In: 52nd IEEE
  Conference on Decision and Control. IEEE, pp. 1930--1937.

\bibitem[{Gu{\c{t}}{\u{a}} and Yamamoto(2016)}]{Guta2}
Gu{\c{t}}{\u{a}}, M., Yamamoto, N., 2016. System identification for passive
  linear quantum systems. IEEE Trans. Autom. Control. 61~(4), 921--936.

\bibitem[{Hayden et~al.(2014)Hayden, Yuan, and Gon{\c{c}}alves}]{nerve}
Hayden, D., Yuan, Y., Gon{\c{c}}alves, J., 2014. Network reconstruction from
  intrinsic noise: Minimum-phase systems. In: 2014 American Control Conference.
  IEEE, pp. 4391--4396.

\bibitem[{HO and Kalman(1966)}]{Ho}
HO, B., Kalman, R.~E., 1966. : Effective construction of linear state-variable
  models from input/output functions. at-Automatisierungstechnik 14~(1-12),
  545--548.

\bibitem[{Hudson and Parthasarathy(1984)}]{path2}
Hudson, R.~L., Parthasarathy, K.~R., 1984. Quantum ito's formula and stochastic
  evolutions. Comm. Math. Phys 93~(3), 301--323.

\bibitem[{James et~al.(2008)James, Nurdin, and Petersen}]{JNP}
James, M.~R., Nurdin, H.~I., Petersen, I.~R., 2008. Control of linear quantum
  stochastic systems. IEEE Transactions on Automatic Control 53~(8),
  1787--1803.

\bibitem[{Kalman(1963)}]{kalman}
Kalman, R.~E., 1963. Mathematical description of linear dynamical systems. SIAM
  1~(2), 152--192.

\bibitem[{Koga and Yamamoto(2012)}]{Naoki}
Koga, K., Yamamoto, N., 2012. Dissipation-induced pure gaussian state. PRA
  85~(2), 022103.

\bibitem[{Levitt and Gu{\c{t}}{\u{a}}(2016)}]{levitt}
Levitt, M., Gu{\c{t}}{\u{a}}, M., 2016. Identification of siso quantum linear
  systems. arXiv preprint:1608.01227.

\bibitem[{Ljung(1987)}]{Ljung}
Ljung, L., 1987. System identification for the user. Englewood Cliffs.
  Prentice-Hall, New Jersey 9, 1213--1225.

\bibitem[{Nurdin(2014)}]{squin}
Nurdin, H.~I., 2014. Quantum filtering for multiple input multiple output
  systems driven by arbitrary zero-mean jointly gaussian input fields. Russian
  J. Math. Phys 21~(3), 386--398.

\bibitem[{Nurdin and Gough(2014)}]{memory2}
Nurdin, H.~I., Gough, J.~E., 2014. Modular quantum memories using passive
  linear optics and coherent feedback. arXiv preprint:1409.7473.

\bibitem[{Nurdin et~al.(2016)Nurdin, Grivopoulos, and Petersen}]{cascade}
Nurdin, H.~I., Grivopoulos, S., Petersen, I.~R., 2016. The transfer function of
  generic linear quantum stochastic systems has a pure cascade realization.
  Automatica 69, 324--333.

\bibitem[{Nurdin et~al.(2009{\natexlab{a}})Nurdin, James, and Doherty}]{nurdin}
Nurdin, H.~I., James, M.~R., Doherty, A.~C., 2009{\natexlab{a}}. Network
  synthesis of linear dynamical quantum stochastic systems. SIAM Journal on
  Control and Optimization 48~(4), 2686--2718.

\bibitem[{Nurdin et~al.(2009{\natexlab{b}})Nurdin, James, and
  Petersen}]{nurdin2009}
Nurdin, H.~I., James, M.~R., Petersen, I.~R., 2009{\natexlab{b}}. Coherent
  quantum lqg control. Automatica 45~(8), 1837--1846.

\bibitem[{Parthasarathy(2012)}]{path}
Parthasarathy, K.~R., 2012. An introduction to quantum stochastic calculus.
  Springer Science \& Business Media.

\bibitem[{Petersen(2016)}]{peter}
Petersen, I.~R., 2016. Quantum linear systems theory. arXiv
  preprint:1603.04950.

\bibitem[{Pintelon and Schoukens(2012)}]{Pintelon&Schoukens}
Pintelon, R., Schoukens, J., 2012. System identification: a frequency domain
  approach. John Wiley \& Sons.

\bibitem[{Stockton et~al.(2004)Stockton, van Handel, and Mabuchi}]{STOCK}
Stockton, J.~K., van Handel, R., Mabuchi, H., 2004. Deterministic dicke-state
  preparation with continuous measurement and control. PRA 70~(2), 022106.

\bibitem[{Tian(2012)}]{Tian}
Tian, L., 2012. Adiabatic state conversion and pulse transmission in
  optomechanical systems. PRL 108~(15), 153604.

\bibitem[{Walls and Milburn(2007)}]{wall}
Walls, D.~F., Milburn, G.~J., 2007. Quantum optics. Springer Science \&
  Business Media.

\bibitem[{Wiseman and Milburn(2009)}]{MILB}
Wiseman, H.~M., Milburn, G.~J., 2009. Quantum measurement and control.
  Cambridge University Press.

\bibitem[{Wolf(2008)}]{WOLF}
Wolf, M.~M., 2008. Not-so-normal mode decomposition. PRL 100~(7), 070505.

\bibitem[{Yamamoto(2014)}]{memory}
Yamamoto, N., 2014. Decoherence-free linear quantum subsystems. IEEE Trans.
  Autom. Control. 59~(7), 1845--1857.

\bibitem[{Yanagisawa and Kimura(2003)}]{Yanagisawa}
Yanagisawa, M., Kimura, H., 2003. Transfer function approach to quantum
  control-part i: Dynamics of quantum feedback systems. IEEE Trans. Autom.
  Control. 48~(12), 2107--2120.

\bibitem[{Youla(1961)}]{Youla}
Youla, D., 1961. On the factorization of rational matrices. IRE Trans. Inf.
  Theory 7~(3), 172--189.

\bibitem[{Zhou et~al.(1996)Zhou, Doyle, Glover, et~al.}]{zhou}
Zhou, K., Doyle, J.~C., Glover, K., et~al., 1996. Robust and optimal control.
  Vol.~40. Prentice hall New Jersey.

\end{thebibliography}

\end{document}